\documentclass[journal]{IAENGtran}

   \usepackage[outdir=./]{epstopdf}
   \usepackage[pdftex]{graphicx}
   \DeclareGraphicsExtensions{.pdf}

\usepackage[cmex10]{amsmath}
\usepackage{amsthm}
\usepackage{amssymb}

\begin{document}
%
\title{Multitape automata and finite state transducers with lexicographic weights}
%
%
%

\author{Aleksander Mendoza-Drosik}

\newtheorem{theorem}{Theorem}
\newtheorem{definition}{Definition}

\maketitle

\pagestyle{empty}
\thispagestyle{empty}

\begin{abstract}
Finite state transducers, multitape automata and weighted automata have a lot in common. By studying their universal foundations, it's possible to discover new insights into all of them. The main result presented here is the introduction of lexicographic finite state transducers, that could be seen as intermediate model between multitape automata and weighted transducers. Their most significant advantage is being equivalent but often exponentially smaller than even smallest nondeterministic automata without weights. Lexicographic transducers were discovered by taking inspiration from Eilenberg's algebraic approach to automata and Solomonoff's treatment of a priori probability. Therefore, a quick and concise survey of those topics is presented, prior to introducing lexicographic transducers. 

\end{abstract}

\begin{IAENGkeywords}
Mealy machines, transducers, sequential machines, computability, complexity
\end{IAENGkeywords}

%
\IAENGpeerreviewmaketitle

\section{Introduction}

\subsection{Preliminaries}

Product of sets $B$ and $C$ is the set $B\times C$ of all ordered pairs $(b,c)$ such that $b\in B$ and $c \in C$. A (partial) function $B\rightarrow C$ is a subset of $B\times C$ such that $(b,c),(b,c)'\in B\rightarrow C$ implies $c=c'$. Given some function $A \subset B\rightarrow C$, we say that $A$ is total if for every $b$ there exists some $c$, such that $(b,c)\in A$. We shall not differentiate between $(B\times C) \times D$ and $B\times (C \times D)$.  We also assume that $\rightarrow$ binds weaker than $\times$, hence $B \times C \rightarrow D$ stands for $(B \times C) \rightarrow D$. One can easily check that $B \rightarrow (C \rightarrow D)$ is the same as $B \times C \rightarrow D$, but it's different from $(B \rightarrow C) \rightarrow D$.

Suppose $\diamond$ is some total function $\diamond \subset A \times A \rightarrow A $, then set $A$ together with $\diamond$ is called a monoid if two criteria are met.
First there must exist some element $1_A \in A$ (called \textit{
	identity element}) such that $\diamond(1_A,a)=\diamond(a,1_A)=a$ for all $a\in A$. Second, it must always hold that $\diamond(\diamond(a_1,a_2),a_3)=\diamond(a_1,\diamond(a_2,a_3))$.
Instead of writing $\diamond(a_1,a_2)$, one can also use infix notation $a_1 \diamond a_2$. Thanks to the second criterion, the order of brackets doesn't matter and we can omit them, as in $a_1 \diamond a_2 \diamond a_3$.

If $A$ contains two elements $a_1$ and $a_2$ such that $a_1 \diamond a_2 = 1_A$, then we call them invertible. $a_2$ can be denoted as $a_1^{-1}$ and called the \textit{inverse} of $a_1$. Monoid, in which every element has some inverse, is called a group. 

If $B\times C$ is a monoid, then $B$ and $C$ must be monoids themselves, with $1_{B\times C} = (1_B,1_C)$. This is called \textit{direct product} of monoids.

Basic understanding of measure theory is assumed.

\subsection{Algebraic foundations}

Suppose  $A$ is some set of \textit{labels}, $Q$ is set of \textit{vertices} and   $\delta \subset Q \times A \times Q$ a set of \textit{edges}. Then $(Q,A,\delta)$ is a labelled directed graph. Define \textbf{path}  to be a finite sequence of edges $(q_{k_1},x_1,q_{k_2}), (q_{k_2},x_2,q_{k_3}),  ... (q_{k_m},x_m,q_{k_{m+1}})$ where $q_{k_i} ,q_{k_{i+1}}\in Q$, $x_i\in A$ and $(q_{k_i},x_i,q_{k_{i+1}}) \in \delta$ for every index $i$.

If $A$ together with operation $\cdot$ (which we call "multiplication") is a monoid, then define \textbf{signature} \cite{PIN} of a path as the result of multiplying consecutive labels  $x_1 \cdot x_2 \cdot ... \cdot x_m $.

\textbf{Automaton} \cite{EILENBERG} is defined as tuple $(Q,I,A,\delta,F)$ where $Q$ and $\delta$ are finite, $A$ is finitely generated and both $I$ and $F$ are subsets of $Q$. It's common to refer to elements of $Q$ as \textit{states}, instead of vertices. Similarly  $A$ is called set of \textit{strings} or \textit{words} instead of labels. All elements belonging to some (usually fixed and known from context) generator of $A$ are called \textit{symbols} or \textit{letters}. Elements of $\delta$ are called \textit{transitions} instead of edges. States that belong to $I$ are called \textit{initial} and those belonging to $F$ are \textit{final}. Sometimes  $\epsilon$ is used instead of $1_A$ to put emphasis that neutral element is an \textit{empty string}.

Path is \textbf{accepting} if it starts in some initial $q_{k_1}$ and ends in final $q_{k_{m+1}}$. An automaton accepts string $x\in A$ if it is a signature of some accepting path. 

The elements of $A$ need not be "actual strings". For instance they might be pairs or triples of elements from other sets. In cases when $A=B\times C$, the automaton is said to be \textbf{multitape}. Note that $B$ or $C$ itself might be nested product of other sets. When $A=(B_1 \times B_2) \times C$, then automaton has 3 tapes. The distinction between single-tape and multitape automata is blurry. Indeed, a pair of letters $(b,c)$ could always be encoded as a single letter $a_{bc}$ (if $B$ has $n$ letters and $C$ has $m$, then $B\times C$ has $n\cdot m$), hence one tape can be used to encode multiple other tapes within. 

There is not much distinction between $A$ and $A\times \{\epsilon\}$. Tape that can only read an empty string, isn't read at all and doesn't make any difference to overall computation. The singleton set $\{\epsilon\}$ is a \textbf{trivial} tape. Usually there is also not much difference between $A=B\times C$ and $A=C \times B$. The order of tapes can be switched and a nearly identical automaton can always be built.

Given $A=B\times C$, the automaton is said to be \textbf{sequential up to} $B$ if $b\ne 1_B,b'\ne 1_B$  implies $(q,(bb',c),q')\notin\delta$. This ensures that as consecutive symbols are read from input tape, the transition that wasn't taken before doesn't suddenly "become valid". Automaton is \textbf{sequential} if it is sequential up to entire $A$. (Note that $A$ is same as $A\times \{\epsilon \}$). For instance, automata that allow entire strings on their edges, are not sequential, white automata that only allow individual symbols, are sequential.

Automaton with $A=B \times C$ is \textbf{deterministic up to} B   if it is sequential up to $B$ and $\vert I \vert = 1$ and $\delta \subset Q \times (B\backslash \{1_B\}) \rightarrow C \times Q$ (when the function is partial, then automaton is called \textbf{partial}, otherwise it's called \textbf{complete}).  Automaton is \textbf{deterministic}  when it is deterministic up to entire A.

Automaton with $A=B \times C$ is \textbf{$\epsilon$-free up to} B if $\delta$ is a subset of $Q \times (B\backslash \{1_B\}) \times C \times Q$. Automaton is \textbf{$\epsilon$-free} if it is $\epsilon$-free up to entire A. 

\textbf{Monoidal language}\cite{mihov_schulz_2019} is any subset of $A$. If $A$ is a free monoid then its subset is called a \textbf{classical language}. Monoidal and classical languages are jointly known under the name of \textbf{formal languages} or simply \textit{languages} for short. Language $L$ is \textbf{rational} if and only if there exists some automaton accepting all strings in $L$ and rejecting all those not in $L$. If $A$ is a direct product of several monoids, then we call $L$ a \textbf{rational relation}. If rational relation is a function, then  automaton recognizing it is called \textbf{functional}.  If $M=(Q,I,A,\delta,F)$ is some automaton, then $\mathcal{L}(M)$ is used to denote language recognized by $M$. If $\mathcal{L}(M)$ is a relation $B\times C$, then  $M(b)$ is used to denote all $c$ such that $(b,c)$ is in $\mathcal{L}(M)$. If automaton has 3 tapes, say $A=B\times C \times D$, then $M(b)$ treats it like $B \times (C \times D)$. Similarly $M(b,c)$  treats $A$ as if it was $(B\times C)\times D$. Moreover, because order of tapes makes little difference, it can be implicitly switched, hence $M(b,d)$ denotes accepted subset of $(B\times D) \times C$.

Automaton with $A=B\times C$ is \textbf{k-valued up to $B$} if the number of outputs $M(b)$ for any $b$ is bounded by constant $k$ (precisely $\vert M(b) \vert \le k$). Functional automata are exactly those that are 1-valued. Automaton is \textbf{k-ambiguous up to $B$} if for every accepted $b$ there are at most $k$ distinct accepting paths with signature $b$. Ambiguous automaton may still be functional if all the accepting paths generate the same outputs. It's possible to decide functionality of automaton in polynomial time\cite{Marie-Pierre}\cite{Gurari}.

The notion of automata can be generalized to \textit{syntactic transformation semigroup}\cite{EILENBERG2}. The $\delta$ function can be seen as the (right) action $Q\times A\rightarrow Q$ of monoid $A$ on arbitrary (possibly infinite) set $Q$. Every element $a$ of $A$ determines some (partial) function $a:Q\rightarrow Q$. When $A$ is a subset $Q \rightarrow Q$, then $(Q,A)$ is called the \textbf{transformation semigroup}. This generalizes the notion of states and alphabet symbols.

Subset $L$ (language) of strings $A=B\times C$ is \textbf{prefix free up to $B$} if there is no element $b_1$ that would be a prefix of another  $b_2$ (the notion of string prefix makes the most sense when $B$ is a free monoid). More formally if $(b_1,c_1)$ and $(b_2,c_2)$ are both in $L$ and $b_1$ is a prefix of $b_2$, then $b_1=b_2$. Automaton is \textbf{subsequential up to $B$} if it is sequential up to $B$ and $\mathcal{L}(M)$ is prefix free up to $B$. This definition is very different from those found in other papers\cite{MOHRI}\cite{MOHRI2}\cite{de_la_higuera}. Usually most authors extend their automata with additional output function for accepting states. If automaton ends in that state, then some additional final output is appended before accepting. Such functionality can be emulated by adding special symbol $\#$ as end marker \cite{HANSAN}. If seen from the perspective of transformation semigroup, all strings and symbols are functions, therefore the end marker is in a sense the same as "state output function". Moreover, any language with end marker is indeed prefix free. The resemblance is analogical to that between \textit{plain kolmogorov complexity} and \textit{prefix free complexity}\cite{KOLMOGOROV}. In essence, sequential automata continue working as long as there is input to read, whereas subsequential automata can "decide on their own" when input should end and can take some additional action.
It's easy to prove that any subsequential machine on minimal number of states can have at most one accepting state. Automata with "state output function" also have such unique accepting state but it's "secretly hidden" in the definition of automaton, rather than explicitly specified in $Q$.

There is no formal distinction between input and output tapes. For instance, given $\mathcal{L}(M) \subset B \times C \times D$, the tape $B$ could be seen as input and $C\times D$ to be the output, when $M(b)$ is used. If instead  $M(b,c)$ is used, then $B\times C$ become input tapes and $D$ becomes output. In case of $M(b,c,d)$ all tapes are input. Automata having 2 tapes, with first one designated as input, are often called \textbf{transducers}.

Norm $\vert\cdot\vert$ is a function that assigns real number to every element of some set. If $A$ is a free monoid, then  define the norm $\vert a \vert$ to be length of string $a$. If $A$ is not free, then it's much less obvious what the length should be. (For instance, if $a_1a_2=a_2$, then $a_1a_2a_2a_2$ might have length 4 or it might have length $2$ because $a_1a_2a_2a_2=a_1a_2a_2=a_1a_2$). When $A$ is a direct product of several free monoids, then one can study the relationship between their lengths. The most notable property is that if $A=\Sigma^*\times\Gamma^*$ and $\delta \subset Q\times \Sigma\times\Gamma\times Q$, then for every accepted $(\sigma,\gamma)\in A$ the lengths  $\vert\sigma\vert$ and $\vert\gamma\vert$ are equal. This can be further generalized to
$A=\Sigma^*_1\times\Sigma^*_2\times...\Sigma^*_n$. If $\delta$ is of the form such that at least one $\Sigma^*_i$ is required to be of length exactly 1 on each transition (that is $\delta\subset Q\times\Sigma^*_1\times...\times\Sigma_i\times...\Sigma^*_n \times Q$) then on the accepted subset of $A$ define induced norm $\vert (\sigma_1,...,\sigma_i,...\sigma_n) \vert = \vert\sigma_i\vert$. Such norm also coincides with length of accepting path, therefore it allows to generalise and apply pumping lemma to multitape automata.

If $A$ contains some elements with inverses then the automaton cannot be sequential. In particular suppose $a a^{-1}=1_A$ and $(q,a',q')\in\delta$ then $a'=a'1_A=a'(aa^{-1})=(a'a)a^{-1}$ but $a'a\ne1_A$ and $a^{-1}\ne1_A$, hence sequentiality is violated. 

Every input tape can be seen as a read-only tape and every output tape can be thought of as write-only. Just as there is no formal distinction between input and output, there is no distinction between read-only and write-only. The difference becomes significant only when we allow read-write tapes, also known as \textit{stacks}. In particular, if $B$ is a group, then pushing $b$ onto stack is the same a reading $b$ from tape. Popping $b$ off of the stack can be seen as reading $b^{-1}$.

$A$ may or may not contain commuting elements. If $A=B\times C$ then all  the elements of the form $(1_B,c)$ and $(b,1_C)$ commute. This phenomenon characterizes nonsequential machines which are used to encode conurrent systems (see theory of traces \cite{DIEKERT}). For this reason, all multitatpe automata with $\epsilon$-transitions are in a sense "concurrent" machines.

\textbf{Configuration} is defined to be a subset of $Q$. 
Given configuration $K$ and $x\in A$, define $\hat{\delta}$ to be transitive closure of $\delta$, that is, $\hat{\delta}(K,x)$ is the set of all states $q$, for which there exists a path starting in $K$ and ending in $q$ with signature $x$. 

In cases when $A=B\times C$  the concept of configuration can be extended to include $C$, that is, define \textbf{superposition} as a subset of $Q \times C$. Given some superposition $S$ define $\hat{\delta}_C$ such that $(q',yy')\in \hat{\delta}_C(S,x)$ whenever there exists $(q,y)\in S$ and path starting in $q$ and ending in $q'$ with signature $(x,y')$. 

Configuration is a way of capturing what states are "active" at a particular moment of computation. Superposition keeps track outputs associated with a each "active" state. Every element of superposition represents one possible branch of nondeterministic computation and the output accumulated along the way.

If some automaton $M$ is sequential up to $B$ and $\epsilon$-free up to $B$, then we can define image of configuration
\[
\delta_C(K,b) = \{  q' \in Q :  \exists_{q\in K} (q,(b,c'),q') \in \delta \}
\]
and image of superposition
\[
\delta_C(S,b) = \{  (q',cc') \in Q \times C:  \exists_{(q,c)\in S} (q,(b,c'),q') \in \delta \}
\]
and then $\hat{\delta}_C$ becomes 
\begin{align*} 
\hat{\delta}_C(S,\epsilon) &=  S \\
\hat{\delta}_C(S,bx)&=\hat{\delta}_C(\delta_C(S,b),x)
\end{align*}
In all four equations above $b$ is an element of smallest generator of $B$ (which is usually fixed and known from context). This gives an effective way of computing the output of $M$, that is, for all $x$ in $B$ there is $c \in M(x)$ whenever $(q,c) \in \hat{\delta}_C(I\times\{1_C\},x)$ for some final state $q$.

\begin{theorem}[Deterministic superposition]
	\label{superposition_deterministic}
	If  automaton over $A=B\times C$ is deterministic up to $B$ then $\vert\hat{\delta}_C(S,x)\vert\le 1$ for all $x\in B$ and all initial superpositions $\vert S \vert = 1$.
\end{theorem}
\begin{proof}
	Determinism states that $\delta \subset Q \times B \rightarrow C \times Q$, so there is at most one transition that can be taken at each step.  Therefore the number of elements in superposition cannot increase.
\end{proof}
As direct consequence it can be shown that in deterministic automata initial state and signature uniquely determine path.
This leads to introduction of the following theorem.

\begin{theorem}[Preservation of prefixes]
	\label{matching_prefixes}
	Let $M$ be some automaton over $A=B\times C$ deterministic up to $B$. For all strings $x,x'\in B$ if $M(xx')=y'\ne\emptyset$ and $M(x)=y\ne\emptyset$, then $y$ is a prefix of $y'$.
\end{theorem}
\begin{proof}
	It follows directly from uniqueness of path that corresponds to signature $xx'$.  
\end{proof}

Theorem \ref{superposition_deterministic} also applies to single-tape automata, because $A$ can be treated like $A\times \{\epsilon\}$. Superposition belonging to $Q\times\{\epsilon \}$ is the same as configuration.

\begin{theorem}[Infinite superposition]
	\label{superposition_infinite}
	Let $M$ be an automaton over $A=B\times C$ sequential up to $B$.  $\vert M(x)\vert=\infty$ for some $x\in B$ only if $M$ contains $\epsilon$-cycle $(q_{k_1},(1_B,y_1),q_{k_2})$,...,$(q_{k_m},(1_B,y_m),q_{k_1})$ where $y_i \in C$ and $(1_B,y_1...y_m)\ne 1_{A}$.
\end{theorem}
\begin{proof}
	Every time a non-$\epsilon$-transition from $\delta\subset Q \times (B \backslash \{1_B\}) \times C \times Q$ is taken, it increases the length of $x$ in the corresponding signature $(x,y)\in A$. Only $\epsilon$-transitions of the form $\delta\subset Q \times \{1_B\} \times C \times Q$ do not increase length of $x$. There are only finitely many elements $y$ of specific finite length $\vert y \vert$. Therefore in order to obtain infinite subset of $C$ it must contain strings of unbounded length. The only way to have unbounded $y$, while keeping $x$ bounded is by taking infinitely many $Q \times \{1_B\} \times C \times Q$ transitions. If there is no $\epsilon$-cycle then only finite number of $\epsilon$ transitions can be taken, before having to take some non-$\epsilon$-transition. Therefore there must be an $\epsilon$-cycle.
\end{proof}

\begin{theorem}[Functional superposition]
	\label{superposition_functional}
	Let $M$ be a functional automaton over $A=B\times C$, sequential up to $B$ and whose recognized language is of the form $L \subset B\rightarrow C$. Then there exists an equivalent automaton such that $\hat{\delta}_C(S,x) \subset Q\rightarrow C$ for all $S\subset Q\rightarrow C$ and $x\in B$ .
\end{theorem}
\begin{proof}
	Suppose to the contrary that there is $x$ and $q$ for which $\hat{\delta}_C(S,x)$ returns relation $Q\times C$ that is not a function $Q\rightarrow C$. Then there are two possibilities: either there is a path that starts in $q$ and ends in $F$ or there is not. If the first case is true, then $M$ is not functional, because we might follow that path and accept with multiple $C$ outputs. If the second case applies, then the state $q$ is redundant and we are free to delete it. 
\end{proof}

\subsection{Stochastic languages and weighted automata}

Suppose that $B$ is a tape of some automaton. If $B$ is a complete semiring then it is called the tape of \textit{weights} and the automaton itself is \textit{weighted}. Completeness is required because infinite sum may arise, although, this requirement can be relaxed for $\epsilon$-free automata (theorem \ref{superposition_infinite}).

\textbf{Probabilistic automaton} is any automaton, whose $A$ is a measure space with total measure $\mu(A)$ equal $1$. Every measurable subset of $A$ is called a \textbf{stochastic language} and can be treated like a random event.


Now a way of constructing automata with probabilistic weights can be presented. Let $N\ge 1$ be some natural number. Take the segment $(0,1)$ of real number line and split it into $N$ equally sized intervals. Let $\Omega=\{\omega_0,\omega_1,\omega_2,...\omega_N\}$ be the set of all those intervals $(\frac{i-1}{N},\frac{i}{N})$, including $\omega_0$ representing $(0,1)$.  Set $\Omega$ generates a monoid with multiplication $\omega_x\cdot \omega_y$ defined as 
\[
\omega_x \cdot \omega_y = (x_0,x_1)\cdot \omega_y=x_0+(x_1-x_0)\cdot\omega_y
\]  
In other words, $\omega_x$ determines linear transformation that treats $\omega_x$ as the new unit  interval $(0,1)$ and $\omega_y$ is made relative to it (for instance, if $\omega_1=(0,0.5)$ and $\omega_2=(0.5,1)$ then $\omega_1\omega_2=(0.25,0.5)$).
Norm $\vert\omega_i\vert$ is equal to the length of interval. It holds that $\vert\omega_x\cdot \omega_y\vert=\vert\omega_x\vert\cdot \vert\omega_y\vert$. Define complete semiring $B$ generated by $\Omega$ with union of intervals as additive operation (hence $\omega_1 + \omega_1\omega_2 = \omega_1$). Norm of $b$ is equal to summing and multiplying norms of individual elements of $\Omega$ (note $\vert b_1 + b_2 \vert \ne \vert b_1 \vert + \vert b_2 \vert$). As $N$ approaches $\infty$, the accuracy of $\Omega$ increases and their sums can approximate any real number. Consider $\Omega_b$ to be the set of all infinite strings starting with $b$ and $\Omega_{\epsilon}$ is the set of all possible infinite strings. The set $B$ can be turned into a measure space by mapping every $b$ into the corresponding measurable set $\Omega_b$. Such definition of measure space corresponds to Solomonoff's \textit{a priori prefix complexity}\cite{KOLMOGOROV}. Norm $\vert b \vert $ coincides with measure  $\mu(\Omega_b)$. For any subset $B'$ of $B$ the measure of $B'$ is equal to the sum $\mu(\Omega_{B'})= \vert \sum\limits_{b\in B'} b \vert$. Note that the subset $\{\omega_1,...,\omega_N\}$ itself has uniform distribution but if $(0,1)$ was partitioned in some irregular way, different distributions could be obtained. Moreover this subset can be seen as a random variable and every sequence of random variables "falls into" some $b$ in $B$ with probability $\vert b \vert$. 

Consider automaton $M$ with single initial state and transitions of the form $Q\times \Omega\backslash\{\omega_0\} \rightarrow C \times Q$. For any input $c$ take the set $M(c) \subset B$ and turn it into prefix-free set $B'$ (that is, if $b_1,b_2\in M(c)$ and $b_1$ is a prefix of $b_2$, then don't include $b_2$ in $B'$). Such set is a random event with probability $P(c)=\mu(\Omega_{B'})$. To prove that $P(c)$ never exceeds $1$, notice that in every prefix-free subset of $B$, no segments of $(0,1)$ overlap, so they can be summed without double-counting. This also implies that $\vert \sum\limits_{b\in B'} b \vert =  \sum\limits_{b\in B'} \vert b \vert $. Every string $b$ uniquely determines some path, so if both $(c,b_1)$ and $(c,b_2)$ belong to $\mathcal{L}(M)$ but $b_1$ is prefix of $b_2$, that means there is $\epsilon$-cycle starting and ending in some final state (so it's only natural and intuitive to discard $b_2$ when counting $P(c)$). 

If automaton has transitions of the form $Q\times  (\Omega\backslash\{\omega_0\}) \rightarrow C \times D \times Q$, then probability of any output $D$ can be calculated for a given input $C$. Probability $P(c,d)$ is the same as $\mu(\Omega_{M(c,d)})$ and $P(c)$ equals sum $\sum_{d\in D}P(c,d)$ of all possible outputs $d$. Then the conditional probability $P(d|c)$ is obtained from $\frac{P(c,d)}{P(c)}$.

The construction described above is called the \textbf{probabilistic semiring}.
Those familiar with the theory of weighted automata\cite{DROSTE}\cite{DROSTE2} might notice that this definition is completely different from the "standard" one. No \textit{formal power series}\cite{SALOMAA} or weight function for transitions\cite{DROSTE2} were used. Apart from assuming that $B$ is a measure space, the definition of automata wasn't extended in any way. Perhaps, the most significant difference is that everything was defined in terms of formal languages and strings, instead of resorting to summation over all possible paths. This presents an alternative approach to weighted automata, that lies much closer to theory formal languages. Tropical semiring (and all others) can be introduced in a similar approach.

Suppose that $A=B\times C$ (the order doesn't matter much) and $C$ is a complete semiring. Given some language $L$  introduce \textbf{quotient} of $L$ denoted with $L\backslash B$ and defined as
\[
(c,b) \in L\backslash B \iff b = \sum_{(b',c)\in L} b'
\]
$L$ can be any subset of $A$ but $L\backslash B$ is specifically a function $C\rightarrow B$. In case of probabilistic semiring, the probability $P(c)=\mu(\Omega_{M(c)})$ is the same as $P(c)=\vert (\mathcal{L}(M)\backslash B )(c) \vert $. This will be the starting point for defining tropical semiring in terms of strings and languages.

Consider automaton $M$ over $A=B \times C$, where $\le_B$ is some relation of total order on $B$. Then $B$ can be turned into semiring with $max$ (or $min$) as additive operation. Hence $B$ can be treated as tape of weights.  This should be called \textbf{max semiring} (or \textbf{min semiring}). If additionally $B$ commutes under multiplication, then it can be called \textbf{arctic semiring} (or \textbf{tropical semiring}). In other papers\cite{MOHRI2}\cite{MOHRI}\cite{DROSTE}, $B$ is required to represent real numbers, but such assumption is very restricting and would require infinitary alphabets\cite{MEER}. If $B$ is a free monoid with $\le_B$ representing lexicographic order, then $B$ is a special case of max semiring (min semiring), called \textbf{lexicographic arctic semiring} (or \textbf{lexicographic tropical semiring}). The lexicographic order itself might be defined by comparing strings from left to right or right to left. Because each time the automaton takes the transition, the weight is appended, rather than prepended, it makes more sense to consider right-to-left order (otherwise only the first transition would matter and the remaining steps of computation would be of little relevance). Therefore this paper considers definition
\[
b_1w_1 > b_2w_2 \iff w_1 > w_2 \mbox{ or }( w_1=w_2\mbox{ and } b_1 > b_2)
\]
where $w_1,w_2$ belong to generator of $B$.
This semiring is a new discovery, which will be investigated in depth in the next part of this paper.

Let $M=(Q,I,C\times D,\delta,F)$ be some automaton that may or may not be deterministic. Define  $\delta' \subset Q\times B\times C\rightarrow D \times Q$ to be a \textbf{disambiguation up to}  $C$ for $M$ if $(Q,I,B\times C \times D,\delta',F)$ is deterministic and $\delta$ coincides with $\delta'$ in the following sense: 
\[
(q,(c,d),q')\in\delta \iff \exists_{b\in B} (q,(b,c,d),q')\in\delta'
\] 
Note that weighted automata can often be seen as disambiguations of some otherwise nondeterministic automata.

Given any $L\subset C \times D \rightarrow B$ maximization of $B$ with respect to $D$ written as $\max\limits_{D\rightarrow B} L$ is defined as
\[
(c,d) \in \max_{D\rightarrow B} L \wedge (c,d',b) \in L \implies b \le L(c,d)
\]
Analogically also define minimization $\min\limits_{D\rightarrow B} L$.
Once a quotient of some weighted automaton is obtained, the weights can be completely erased by either minimising or maximizing them. 

Consider automaton over $A=B\times C$ with $C=C_1\times C_2 \times ... C_n$ where every $C_i$ is a max (min) semiring. Then $C$ can be turned into max (min) semiring by treating $C_i$ to the left as "more important" than those to the right. More formally $(c_1,(c_2,...c_n))>(c_1',(c_2',...c_n'))$ if and only if either $c_1>c_1'$ or $c_1=c_1'$ and recursively $(c_2,...c_n)>(c_2',...c_n')$. Such construction of $C$ is known as \textbf{lexicographic semiring} \cite{roark-etal-2011-lexicographic}.


\section{Lexicographic tropical semiring}

Consider automaton $M$ over $A=W^*\times \Sigma^* \times D$ with transitions $\delta \subset Q \times W \times \Sigma \times D \times Q$ and total order $\le_W$, which induces lexicographic order on $W^*$, making it a lexicographic tropical semiring. This guarantees that for any $b\in W^*$, $c \in \Sigma^*$ and $d \in D$ if $(b,c,d)$ is in $\mathcal{L}(M)$, then lengths $\vert b \vert$ and $\vert c \vert$ are equal. For any input $c$ the output $M(c)$ can be computed and after dividing it by $W^*$, the quotient $M(c) \backslash W^*$ is a function $D \rightarrow W^*$ assigning (lexicographically) lowest possible path to every obtainable output $D$. Because $min$ is used as semiring addition, the quotient $M(c) \backslash W^*$ becomes \[
(d,b) \in M(c)\backslash W^* \iff b = \min_{(b',d)\in M(c)} b'
\] and all $b'$ are of equal lengths (same as $\vert c \vert$). Because there are only finitely many strings of any fixed length, there is no need to require $W^*$ to be a complete semiring. Moreover, the order lexicographic $\le_{W^*}$  need not be total because comparison will never occur for strings of different lengths. Such $M$ will be referred to as \textbf{lexicographic transducers}.

An interesting property emerges, when studying superpositions $Q \times W^* \times D$.  Suppose that $S$ is some superposition obtained on lexicographic transducer by reading string $\sigma_1\sigma_2...\sigma_k \in \Sigma^*$. Let $(q,b,d)\in S$ and imagine that the automaton reads next symbol $\sigma_{k+1}$ and enters new superposition $S'$. As it takes some transition $(q,\sigma_{k+1},w,d',q')$, it causes the element $(q',bw,dd')$ to be included in $S'$.  If there were two elements $(q,b_1,d_1)$ and $(q,b_2,d_2)$ in $S$ and $b_1 < b_2$, then the inequality would still be preserved for $b_1w < b_2w$ in $S'$. In that sense, the superpositions are monotonous and  $(q,b_1,d_1)$ can be safely removed from $S$ without making any difference to $M(c)\backslash W^*$.

On the other hand, suppose that $(q_1,b_1,d_1)$ and $(q_2,b_2,d_2)$ are in $S$ and then automaton takes transitions $(q_1,\sigma_{k+1},w_1,d_1',q')$ and $(q_2,\sigma_{k+1},w_1,d_2',q')$ both leading to the same $q'$ over the same $\sigma_{k+1}$. Such states are said to be \textbf{conflicting}. In order to determine whether  $b_1w_1 > b_2w_2$, all that's needed to know is $b_1>b_2$ and $w_1>w_2$ but it's not necessary to know the actual strings, because by definition 
\[
b_1w_1 > b_2w_2 \iff w_1 > w_2 \mbox{ or }( w_1=w_2\mbox{ and } b_1 > b_2)
\]
This introduces everything that's necessary for the following theorem.
\begin{theorem}[Weights can be erased]
	Let $M$ be some lexicographic  transducer over $A=W^*\times \Sigma^*\times D$ then there exists automaton $N$ over $\Sigma^* \times D$ equivalent to  $ \min\limits_{D\rightarrow W^*}\mathcal{L}(M)\backslash W^*$.
	\label{weights_erased}
\end{theorem}
\begin{proof}
	Suppose that $M=(Q,I,A,\delta,F)$ and $N=(Q',I',\Sigma^* \times D,\delta',F')$. Conversion can be carried out using "extended" powerset construction. Instead of using $Q'=2^Q$, which can keep track of current configuration in $Q$, we need to keep track of superposition $Q \times W^*$. However, because there are infinitely many strings $W^*$, such powerset would result in infinite $Q'$. To make $Q'$ bounded, we abstract the exact strings $W^*$  away and only focus on the order relationship between them.  More precisely, let $S\subset Q \rightarrow W^*$ be some superposition and let $\phi_S$ be a formula of the following form:
	\[
	\epsilon < S(q_1) < S(q_2) = S(q_3) < ... < ... = ... < ... = S(q_n)
	\]
	where $q_1...q_n$ are all the states included in given $S$.
	Let $\Phi$ be the set of formulas for all possible superpositions. We can treat $\Phi$ as equivalence classes for $Q \rightarrow W^*$. Note that it's enough to only consider $Q \rightarrow W^*$ instead of $Q \times W^*$, because (as shown a few paragraphs before) the weights are monotonous and we can  remove all but the smallest one. 
	
	Having said all this, put $Q'= Q \times \Phi$. The extra $Q$ is needed, because we want to pick one representative state $q$ from every formula $\phi$. We can immediately remove all those elements $(q,\phi)$ of $Q'$ for which $q$ cannot be found in $\phi$.  The state $q$ we be used to keep track of $D$. Hence superposition $S'$ in $N$ that corresponds to $Q' \times D$ translates to $Q \times \{\phi\} \times D$, which can be seen as entire class of  superpositions $Q \times W^* \times D$. 
	
	The set of states used in any given $\phi$ determines some configuration $K_\phi$. For every two states $(q_1,\phi_1),(q_2,\phi_2)\in Q'$ we put transition from $(q_1,\phi_1)$ to $(q_2,\phi_2)$ over symbol $\sigma$ with output $d$, whenever configuration $K_{\phi_1}$ transitions to $K_{\phi_2}$  over $\sigma$ (formally $\delta_{W^*\times D}(K_{\phi_1},\sigma)=K_{\phi_2}$) and the state $q_1$ itself also transitions to $q_2$ (formally $(q_1,w,\sigma,d,q_2)\in\delta$) and the formula $\phi_2$ indeed holds true (after transitioning from $\phi_1$ over $w$ ). In a moment some of those transitions will need to be removed in order to simulate the effect of erased weights. Before that, we should first add one more extra state $f$ to $Q'$, which will be the only accepting state of $N$. Every time we put transition from $(q_1,\phi_1)$ to $(q_2,\phi_2)$ and $q_2$ is an accepting state, we need to put the exact same transition from $(q_1,\phi_1)$ to $f$. (This way we can simulate $\epsilon$-transition from $(q_2,\phi_2)$ to $f$.) For every initial state $q_i$ of $M$, we designate $(q,\phi)$ as initial state of $N$, where $\phi$ is the formula 
	\[\epsilon = S(q_1) = ... = S(q_i) = ... = S(q_n)\]
	and $\{q_1,...,q_n\}$ is the set of initial states $I$.
	If any initial state is also an accepting state, then we additionally set $f$ as initial state of $N$.
	
	Finally, the last step of conversion is to find all conflicting states and remove the transitions with lower weights. Recall that if there are two states $(q_1,\phi),(q_2,\phi)\in Q'$ transitioning to the same third state $q'_3\in Q'$ over the same symbol $\sigma$, then we call $(q_1,\phi_1)$ and $(q_2,\phi_2)$ conflicting. Remember that every transition in $\delta'$ is a "copy" of some weighted transition in $\delta$ (including those leading to $f$). Let's say that $w_1$ is the weight that "would be" put between $(q_1,\phi)$ and $q_3'$, if we hadn't erased it. Similarly fro $w_2$ and $(q_2,\phi)$. Next we need to lookup if according to $\phi$ the state $q_1$ carries lower or higher weight than $q_2$. This, together with the $w_1$ and $w_2$, gives us enough information to decide which of the transitions should be erased (if any). 
	
	This concludes the construction of $N$. Note that $N$ is nondeterministic and it's not possible to reach such configuration of $Q'$ in which two states would have different $\phi$ (This does not include $f$ which doesn't have any $\phi$ associated with it. There is no problem, because $f$ has no outgoing transitions). 
\end{proof}

Lexicographic transducer $M$ is said to be functional when the relation $ \min\limits_{D\rightarrow W^*}\mathcal{L}(M)\backslash W^*$ is functional. Theorem \ref{superposition_functional} together with theorem \ref{weights_erased} tells that every functional $M$ (after removing dead-end states) has all reachable superpositions  of the form $Q  \times W^* \rightarrow D$. After  removing all but the lowest weight (due to monotonicity), that leaves only $Q  \rightarrow W^* \times D$. This implies that any time there are two conflicting states, either the weights on transitions are different or they are the same but also the associated $D$ outputs are the same. If there was a conflicting pair of states with equal weights but different $D$, that would break functional nature of automaton and lead to ambiguous output.

Suppose that the automaton $M$ has no reachable conflicting states with equal weights and only one single accepting state.  Then $M$ is guaranteed to be functional. Such automata are called \textbf{strongly functional}. The lexicographic tropical semiring becomes "unnecessary" because the formula
\[
b_1w_1 > b_2w_2 \iff w_1 > w_2 \mbox{ or }( w_1=w_2\mbox{ and } b_1 > b_2)
\]
always falls into the left side of "or" and the recursion on the right never happens ($w_1 > w_2$ always holds). Therefore it's not necessary to keep the history of weights $W^*$ in the superposition $Q\rightarrow W^* \times D$. They can be dropped altogether and all computation can be carried out with only $Q\rightarrow D$.  Another special property of such $M$ is that it's "deterministic in reverse", that is, given sequence of configurations $K_0,K_1,...K_n$ for each step $\sigma_1,\sigma_2,...\sigma_n$ of computation, it can backtracked from accepting state back to initial state in a deterministic way, because when given particular state $q_i$ in $K_i$  ($q_n$ is the unique final state) and $\sigma_i$, then there is always only one smallest weight that could lead to $q_i$ from some state of $K_{i-1}$. Moreover, the unique final state is not a limitation, because nondeterminism can be used to simulate $\epsilon$-transitions (similarly to the way is was done in theorem \ref{weights_erased}) or a special end-marker $\#$ could be introduced. Now it can be shown that weights of such automata can be erased in a simpler way than in theorem \ref{weights_erased}.

\begin{theorem}[Weights can be erased - strongly functional case]
	Let $M$ be some lexicographic  transducer over $A=W^*\times \Sigma^*\times D$ that has only one final state and no conflicting states with equal weights. Then there exists automaton $N$ over $\Sigma^* \times D$ equivalent to  $ \min\limits_{D\rightarrow W^*}\mathcal{L}(M)\backslash W^*$.
	\label{weights_erased_simple}
\end{theorem}
\begin{proof}
	Similar to the previous case but this time we put $Q'=2^Q\times Q \cup \{f\}$. If $S$ is configuration in $M$, we know that we only need to keep track of $Q \rightarrow D$, instead of $Q \times W^* \times D$. Therefore we can convert $S$ to superposition $S'$ of $N$, by setting $((K_S,q),d)\in S'$ for every  $(q,d)\in S$, where $K_S$ is the configuration corresponding to $S$.  
	
	We put transition from $(K_1,q_1)$  to $(K_2,q_2)$ in $Q'$ over $\sigma$ with output $d$, whenever $K_1$ transitions to $K_2$ over $\sigma$ (formally $\delta_{W^*\times D}(K_1,x)=K_2$) and $q_1$ transitions to $q_2$ over $\sigma$ and $d$ (formally $(q_1,w,\sigma,d,q_2) \in \delta$). 
	
	We make state $(K,q)$ of $N$ final whenever $q$ is final. The initial states of $N$ are all of the form $(I,q)$ for each $q$ in $I$.
	
	The last step is to find all conflicting states, that is, two states $(K,q_1),(K,q_2)\in Q'$ having the same configuration $K$ and transitioning to some $q'\in Q'$ over the same $\sigma$. The transitions $((K,q_1),\sigma,d,q')$ and $((K,q_2),\sigma,d,q')$ are called conflicting transitions. Every time we encounter them, we delete the one with higher weight. We will find out their weights by looking put what transition from $\delta$ lead to their creation. It will never happen that two conflicting transitions have equal weights.
\end{proof}

Lexicographic transducers (even the strongly functional ones) can be exponentially smaller than the smallest nondeterministic equivalent 2-tape automata. To show this we will need help of Myhill-Nerode theorem.

Let $L\subset B \rightarrow C$ be some (partial) function and let $b_1,b_2\in B$. Element $(b,c)\in A$ is a \textbf{distinguishing extension up to $B$} of $b_1$ and $b_2$ if exactly one of $(b_0b,L(b_0)c)$ or $(b_1b,L(b_1)c)$ belongs to $L$. Define an equivalence relation $=_L$ on $A$ such that $a_0 =_L a_1$ if and only if there is no distinguishing extension up to $B$ for $a_0$ and $a_1$.
\begin{theorem}[Generalized Myhill-Nerode theorem]
	\label{myhill_nerode}
	Let $L\subset B \rightarrow C$. Assume that $L(bb')=c'\ne \emptyset$ and $M(b)=c\ne \emptyset$ implies $c$ is a prefix of $c'$ (preservation of prefixes holds). $L$  can be recognized by automaton deterministic up to $B$ if and only if there are only finitely many equivalence classes induced by $=_L$.
\end{theorem}
\begin{proof}
	($\impliedby$) First assume there are finitely many equivalence classes. Let $G$ be the smallest generator of $B$. Then build an automaton by treating every class as a state $Q$. Put a transition from class $q$ to $q'$ over $b'\in G\backslash 1_B$ whenever there exists $(b,c)\in q$ and $(bb',c')\in q'$. Preservation of prefixes guarantees existence of suffix $s$ such that $cs=c'$. This suffix shall be used as transition output. By taking $b'$ only from $G\backslash 1_B$ we ensure that automaton is sequential up to $B$ and has no $\epsilon$-transitions. The class that contains $1_A$ is designated as the unique initial state (hence the automaton is deterministic up to $B$). All the classes intersecting $L$ are accepting states (note that if $q\cap L \ne \emptyset$ then $q\subset L$).
	
	($\implies$) Conversely, if there is an automaton deterministic up to $B$ and recognizing $L$, then there could be found a homomorphism from states of machine to equivalence classes. (The exact proof is well covered in most introductory courses to automata theory and this generalized version is largely analogical, so we won't elaborate on this proof much further.)
\end{proof}

The above theorem no longer works when automaton is not deterministic at least up to $B$. 
\begin{theorem}
	There exists a family of strongly functional lexicographic transducers such that their equivalent minimal 2-tape nondeterministic automata (after erasing weights) are exponentially larger.
	\label{exp_weighted}
\end{theorem}
\begin{proof}
	Define family of strongly functional lexicograpgic transducers in such a way that for every $i \ge 3 $ there is one defined on $i$ states. The number of states of minimal equivalent 2-tape automaton is $O(2^i)$. Figure \ref{exp_weighted_fig} presents a way to build such automata. State $q_0$ is initial. Using strings from $\{0,1\}^*$ one can obtain any configuration of states $q_1$ to $q_n$. Let's associate each configuration with a string $z\in\{0,1\}^n$ (for instance $z=011$ would be a configuration $\{q_2,q_3\}$). That gives  $2^n$ possible strings. State $q_{n+1}$ is accepting and all the states $q_1$...$q_n$ are connected to it. Essentially the relation described by this automaton is a subset of $\{0,1\}^+2 \times \{y_1,...,y_n\}$. Also suppose that weights $w_1$...$w_n$ are in strictly ascending order. Then the automaton maps every $z$ determined by $x\in\{0,1\}^+$ to some $(x,y_k)$ such that $k$ indicates the least significant bit in $z$. For instance suppose $n=4$ and $x=0011\sim z=1101$ then $(x2,y_4)$,$(x02,y_3)$,$(x002,y_4)$,$(x0002,y_4)$,$(x00002,\emptyset)$. Notice that one can reconstruct $z$ from such sequence of $y$'s. 
	
	Every 2-tape automaton that has disjoin alphabets in each tape can be simulated by a single-tape automaton reading union of those alphabets. In this case such union is $D=\{0,1,2,y_1,...,y_n\}$. This way the language becomes subset of $\{0,1\}^+2\{y_1,...,y_n\}$ and every pair $(x,y_k)$ becomes a string $xy_k$. Using Myhill-Nerode theorem, it can be seen that no two $x$ strings that map to two different $z$ are equivalent, hence the smallest deterministic FSA must have at least $2^n$ states. Call this minimal automaton $\mathcal{A}$. The most difficult problem is to show that no nondeterministic automaton polynomially smaller than $\mathcal{A}$ can be build. The rigorous proof can be obtained with help of Theorem 7 and Lemma 7 presented by Kameda and Weiner \cite{KAMEDA}. We can build RAM using $D(\mathcal{A})$ and $D(\overleftarrow{\mathcal{A}})$ (all defined in \cite{KAMEDA}).  All configurations of states $q_1$...$q_n$  have different succeeding event\cite{KAMEDA} and none of them is subset of the other (because there exists bijection between $z$ and sequence of $y$'s produced by $(x2,y_{k_0})$,$(x02,y_{k_1})$,...). Hence the minimal legitimate grid cannot be extended for any of them and the nondeterministic FSA cannot be much smaller than $2^n$. 
\end{proof}

\begin{figure}[!t]
	\centering
	\includegraphics[width=8cm]{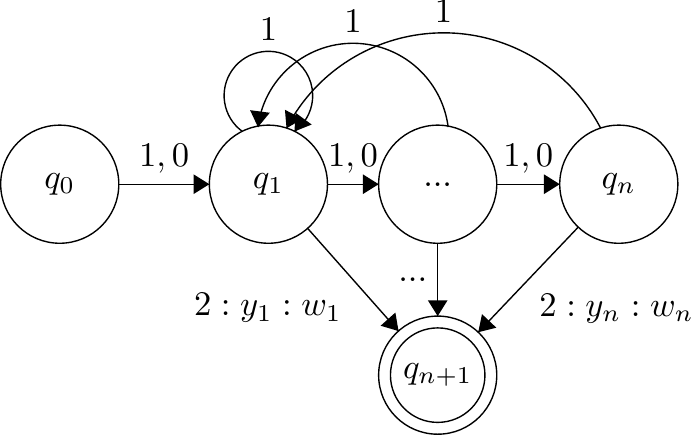} 
	\caption{In this sketch of lexicographic transducer all the weights $w_1$,...,$w_n$ are distinct. In many other transitions, weights were omitted, as they don't play any role and could be arbitrary.}
	\label{exp_weighted_fig}
\end{figure}
One can easily notice that every 2-tape automaton can be treated like a lexicographic transducer with all weights equal, therefore the opposite of theorem \ref{exp_weighted} doesn't hold (there is no family of 2-tape automata such that lexicographic transducers would be larger).

It's possible to decide whether lexicographic transducers are functional using quadratic procedure analogical to the one for unweighted transducers \cite{Marie-Pierre}. In case of strongly functional lexicographic transducers the procedure becomes even simpler, as instead of using "Advance \& Delay"\cite{Marie-Pierre}, it's enough to square automaton and make sure that it has no weight-conflicting transitions. 

When introducing transformation semigroup, the strings in $A$ were treated as partial functions $Q\rightarrow Q$.  Every such function takes some configuration and produces a new one. When $A=B \times C$, superpositions can be used instead. Every string $b$ in $B$ becomes a function $b:Q\times C \rightarrow Q \times C$. In case of nondeterministic automata, there is $b:2^{Q\times C} \rightarrow 2^{Q \times C}$.

In the particular case of strongly functional lexicographic transducers over $A=W^*\times\Sigma^*\times D$, every $x$ in $\Sigma^*$ becomes  a function $x : 2^{Q \rightarrow D} \rightarrow 2^{Q \rightarrow D}$. (Notice how $W$ wasn't included, because the attempt is not to model $\mathcal{L}(M)$ but rather the language after minimization $ \min\limits_{D\rightarrow W^*}\mathcal{L}(M)\backslash W^*$.) This leads to conclusion that  lexicographic transducers are not a specialization, but rather a generalization of  transducers, because weights $W$ are merely a tool to take greater control over the transformation semigroup. This could not be said about other types of weighted transducers, as their transformation semigroups cannot be expressed without also including information about weights accumulated in each superposition (just using $2^{Q \times D}$ without tape of weights would not be enough). 

This can shed some additional light for theorem \ref{exp_weighted}. Even though nondeterministic single-tape automata can have exponentially less states than deterministic ones, the number of all reachable configurations would still be equal in both. The smallest subset $2^Q$ for  transformation semigroup $2^Q\rightarrow2^Q$ would be isomorphic to smallest set $Q$ in deterministic $Q\rightarrow Q$. This situation is different for multitape automata and it's what lexicographic transducers try to take advantage of.

\section{Conclusions}
Weighted automata don't have to be an alien concept that requires any special extensions. All weights can be viewed as  tapes over alphabets with some particular properties. This more general approach give us necessary foundations for defining lexicographic transducers. They were invented by trying to generalize and simplify weighted automata. There is yet a lot to discover. Theorem \ref{exp_weighted} gives certain clues, that perhaps they could be inferred\cite{de_la_higuera} more efficiently, or at least generalise better. Solomonoff's theory of inductive inference \cite{SOLOMONOFF}\cite{KOLMOGOROV}\cite{SOLOMONOFF2} says that simpler and shorter automata, should be the preferred solution to inference problems. Lexicographic transducers can  express complex "replace-all" functions in simpler and more reliable ways than other weighted automata, specifically thanks to lack of commutativity in lexicographic tropical semiring. They seem perfectly suited for tasks that require the automaton to "forget history" of their weights.
\section*{Acknowledgment}

The author would like to thank Piotr Radwan for all the great inspiration.

\ifCLASSOPTIONcaptionsoff
  \newpage
\fi



%




\bibliographystyle{BibTeXtran}   
\bibliography{BibTeXrefs}       

\begin{thebibliography}{10}
\providecommand{\url}[1]{#1}
\csname url@samestyle\endcsname
\providecommand{\newblock}{\relax}
\providecommand{\bibinfo}[2]{#2}
\providecommand{\BIBentrySTDinterwordspacing}{\spaceskip=0pt\relax}
\providecommand{\BIBentryALTinterwordstretchfactor}{4}
\providecommand{\BIBentryALTinterwordspacing}{\spaceskip=\fontdimen2\font plus
\BIBentryALTinterwordstretchfactor\fontdimen3\font minus
  \fontdimen4\font\relax}
\providecommand{\BIBforeignlanguage}[2]{{%
\expandafter\ifx\csname l@#1\endcsname\relax
\typeout{** WARNING: IEEEtran.bst: No hyphenation pattern has been}%
\typeout{** loaded for the language `#1'. Using the pattern for}%
\typeout{** the default language instead.}%
\else
\language=\csname l@#1\endcsname
\fi
#2}}
\providecommand{\BIBdecl}{\relax}
\BIBdecl

\bibitem{PIN}
J.-E. Pin, \emph{Mathematical Foundations of Automata Theory}.\hskip 1em plus
  0.5em minus 0.4em\relax American Mathematical Society, 2017.

\bibitem{EILENBERG}
S.~Eilenberg, \emph{Automata, Languages and Machines Vol. A}.\hskip 1em plus
  0.5em minus 0.4em\relax Academic Press, 1974.

\bibitem{mihov_schulz_2019}
S.~Mihov and K.~U. Schulz, \emph{Finite-State Techniques: Automata, Transducers
  and Bimachines}, ser. Cambridge Tracts in Theoretical Computer Science.\hskip
  1em plus 0.5em minus 0.4em\relax Cambridge University Press, 2019.

\bibitem{Marie-Pierre}
M.-P. B{\'e}al, O.~Carton, C.~Prieur, and J.~Sakarovitch, ``Squaring
  transducers: An efficient procedure for deciding functionality and
  sequentiality of transducers,'' in \emph{LATIN 2000: Theoretical
  Informatics}, G.~H. Gonnet and A.~Viola, Eds.\hskip 1em plus 0.5em minus
  0.4em\relax Berlin, Heidelberg: Springer Berlin Heidelberg, 2000, pp.
  397--406.

\bibitem{Gurari}
I.~O. Gurari, E.M., ``A note on finite-valued and finitely ambiguous
  transducers,'' \emph{Math. Systems Theory}, 1983.

\bibitem{EILENBERG2}
S.~Eilenberg, \emph{Automata, Languages and Machines Vol. B}.\hskip 1em plus
  0.5em minus 0.4em\relax Academic Press, 1976.

\bibitem{MOHRI}
F.~P. Mehryar~Mohri and M.~Riley, ``Weighted finite-state transducers in speech
  recognition,'' \emph{AT\&T Labs – Research}, 2008.

\bibitem{MOHRI2}
M.~Mohri, \emph{Weighted Finite-State Transducer Algorithms. An
  Overview}.\hskip 1em plus 0.5em minus 0.4em\relax Springer, 2004.

\bibitem{de_la_higuera}
C.~de~la Higuera, \emph{Grammatical Inference: Learning Automata and
  Grammars}.\hskip 1em plus 0.5em minus 0.4em\relax Cambridge University Press,
  2010.

\bibitem{HANSAN}
C.~E. Hasan Ibne~Akram, Colin de la~Higuera, ``Actively learning probabilistic
  subsequential transducers,'' \emph{JMLR: Workshop and Conference
  Proceedings}, 2012.

\bibitem{KOLMOGOROV}
N.~V. A.~Shen, V. A.~Uspensky, \emph{Kolmogorov Complexity and Algorithmic
  Randomness}.\hskip 1em plus 0.5em minus 0.4em\relax IRIF, 2019.

\bibitem{DIEKERT}
V.~Diekert, \emph{The Book Of Traces}.\hskip 1em plus 0.5em minus 0.4em\relax
  Wspc, 1995.

\bibitem{DROSTE}
M.~Droste, W.~Kuich, and H.~Vogler, \emph{Handbook of Weighted Automata}, 01
  2009.

\bibitem{DROSTE2}
M.~Droste and D.~Kuske, ``Weighted automata,'' \emph{Institut fur Informatik,
  Universitat Leipzig}, 2010.

\bibitem{SALOMAA}
M.~S. Arto~Salomaa, \emph{Automata-Theoretic Aspects of Formal Power
  Series}.\hskip 1em plus 0.5em minus 0.4em\relax Springer-Verlag New York.

\bibitem{MEER}
K.~Meer and A.~Naif, ``Generalized finite automata over real and complex
  numbers,'' vol. 591, 04 2014.

\bibitem{roark-etal-2011-lexicographic}
\BIBentryALTinterwordspacing
B.~Roark, R.~Sproat, and I.~Shafran, ``Lexicographic semirings for exact
  automata encoding of sequence models,'' in \emph{Proceedings of the 49th
  Annual Meeting of the Association for Computational Linguistics: Human
  Language Technologies}.\hskip 1em plus 0.5em minus 0.4em\relax Portland,
  Oregon, USA: Association for Computational Linguistics, Jun. 2011, pp. 1--5.
  [Online]. Available: \url{https://www.aclweb.org/anthology/P11-2001}
\BIBentrySTDinterwordspacing

\bibitem{KAMEDA}
P.~W. Tsunehiko~Kameda, ``On the state minimization of nondeterministic finite
  automata,'' \emph{IEEE Transactions on Computers}, 1970.

\bibitem{SOLOMONOFF}
R.J.Solomonoff, ``A formal theory of inductive inference. part i,''
  \emph{Information and Control}, 1964.

\bibitem{SOLOMONOFF2}
------, ``A formal theory of inductive inference. part ii,'' \emph{Information
  and Control}, 1964.

\end{thebibliography}

%







\end{document}